\documentclass[]{elsarticle}
\usepackage[sc,noBBpl]{mathpazo}
\usepackage[mathscr]{eucal}
\usepackage{amsmath,amsfonts,amssymb,amsthm}
\usepackage{graphicx}
\usepackage{mathtools}
\usepackage{color}
\usepackage{tgpagella}
\usepackage{subfigure}
\usepackage{caption}
\usepackage{url}

\theoremstyle{plain}
\newtheorem{theorem}{Theorem}
\newtheorem{lemma}{Lemma}

\newtheoremstyle{note}{\topsep}{\topsep}{\slshape}{}{\scshape}{}{ }{}
\theoremstyle{note}

\numberwithin{equation}{section}
\numberwithin{theorem}{section}
\numberwithin{lemma}{section}
\numberwithin{proposition}{section}
\numberwithin{corollary}{section}
\numberwithin{remark}{section}

\mathtoolsset{mathic,centercolon}
\newcommand\scN{{\mathscr N}}

\newcommand\mvector{\boldsymbol}

\newcommand\vq{\mvector{q}}

\newcommand\vJ{\mvector{J}}

\newcommand\vP{\mvector{P}}

\newcommand\vvarphi{\mvector{\varphi}}

\newcommand\field{\mathbb}

\newcommand\R{\field{R}}

\newcommand\C{\field{C}}

\newcommand\N{\field{N}}

\newcommand\rmi{\mathrm{i}\mspace{1mu}}

\newcommand\Dt{\frac{\mathrm{d}\phantom{t} }{\mathrm{d}\mspace{1mu}
t}}

\newcommand\Dz{\frac{\mathrm{d}\phantom{z} }{ \mathrm{d}z}}

\newcommand\Dtt{\frac{\mathrm{d}^2\phantom{t} }{\mathrm{d}t^2}}
\newcommand\Dzz{\frac{\mathrm{d}^2\phantom{z} }{\mathrm{d}z^2}}

\begin{document}
\begin{frontmatter}
\title{Non-integrability of the semiclassical Jaynes--Cummings models
without the rotating-wave approximation}

\author{Andrzej Maciejewski}
\ead{ a.maciejewski@ia.uz.zgora.pl}
\address{Janusz Gil Institute of Astronomy,
  University of Zielona G\'ora,
  Lubuska 2, PL-65-265 Zielona G\'ora, Poland}

\author{Wojciech Szumi\'nski}
\ead{w.szuminski@if.uz.zgora.pl}
\address{Institute of Physics, University of Zielona G\'ora, Licealna 9, PL-65-407, Zielona G\'ora, Poland}

\begin{abstract}
   Two versions  of the semi-classical Jaynes--Cummings model without the rotating wave approximation  are investigated. It is shown that for a  non-zero value of the coupling constant  the version introduced by Belobrov, Zaslavsky, and Tartakovsky is Hamiltonian with respect to a certain degenerated Poisson  bracket. Moreover, it is shown that both models are not integrable.   
\end{abstract}

\begin{keyword}
Semi-classical Jaynes--Cummings models; Poincar\'e cross sections; Non-integrability; Variational equations; Differential Galois group.
\end{keyword}

\end{frontmatter}
\section{Introduction}
The Jaynes--Cummings model describes a system of $n$ two-levels atoms
interacting with a single mode of the electromagnetic field. Its
rotating-wave approximation is exactly solvable. With this
approximation, its semi-classical version is also solvable.  On the
other hand, numerical investigations of the semi-classical version of
the Jaynes--Cummings model without the rotating wave approximation
shows its chaotic behaviour for a large range of its parameters.
Nevertheless, particular periodic solutions of this system were found
analytically.

Till now an integrability analysis of this system was not
performed. Here we investigate this problem in the framework of
differential Galois theory.

There are two version of the the semi-classical Jaynes--Cummings
system. The first of them, investigated in \cite{Belobrov:76} has the
following form
\begin{subequations}
  \label{eq:system 1}
  \begin{align}
    \dot x &=-y  \\
    \dot y &=x +z E ,\\
    \dot z &=-y E , \\
    \dot E &=B ,\\
    \dot B &=\alpha x -\mu^2E ,
  \end{align}
\end{subequations}
where $\mu=\omega/\omega_0$ and $ \alpha=n\mu(2\lambda/\omega_0)^2$
are the dimensionless parameters, and the dot denotes the derivative
with respect to rescaled time $\tau=\omega_0 t$. Here $\omega_0$ is
the transition frequency of each atom and $n$ is the number of two
level atoms. It is easy to verify that the system~\eqref{eq:system 1}
posses two first integrals
\begin{equation}
  \label{eq:bloch}
  F= x^2+y^2+z^2,
\end{equation}
and the energy integral
\begin{equation}
  \label{eq:W}
  H=\alpha z-\alpha xE+\frac{1}{2}\mu^2 E^2+\frac{1}{2}B^2.
\end{equation}
Jele\'nska-Kukli\'nska in~ \cite{Jelenska:90::}, and Kujawski
in~\cite{Kujawski:88::,Kujawski:89}, found exact solutions of the
system~\eqref{eq:system 1}. They can be expressed in terms of elliptic
Jacobian function
\begin{equation}
  \label{eq:sol}
  E=E_0\operatorname{cn}(\Omega t,m),
\end{equation}
where \begin{equation}
  \label{eq:omega}
  \Omega^2=\pm\frac{1}{2}\left(\mu^2-\frac{1}{3}\right)\sqrt{1+\chi},
\end{equation}
\begin{equation}
  m=\frac{1}{2}+\frac{\mu^2-\frac{1}{3}}{4\Omega^2},\quad E_0^2=16m\Omega^2,
\end{equation}
and \begin{equation}
  \label{eq:chi} \chi=\left(\mu^2-\frac{1}{3}\right)^{-2}\left\{\pm\frac{4}{3}\sqrt{\alpha^2-4\left(\mu^2-\frac{1}{9}\right)^3}-\left(\mu^2-\frac{1}{9}\right)\left(\mu^2-\frac{17}{9}\right)\right\}.
\end{equation}
Solution~\eqref{eq:sol} lies on energy level $H((x,y,z,E,B)=h$ with
\begin{equation}
  \label{eq:hh} h=\pm\frac{5}{3}\sqrt{\alpha^2-4\left(\mu^2-\frac{1}{9}\right)^3}-2\left(\mu^2-\frac{1}{9}\right)\left(\mu^2-\frac{5}{9}\right).
\end{equation}
The components of the Bloch vector given in terms of the electric
field $E$ are defined by
\begin{equation}
  \label{eq:xyz}
  \begin{split} x(E)&=\frac{1}{\alpha}\left[\frac{3}{2}\left(\mu^2-\frac{1}{9}\right)E-
      \frac{1}{8}E^3\right],\qquad
    y(E)=-\dot x(E),\\
    z(E)&=
    \frac{1}{\alpha}\left[-\frac{3}{2}\left(\mu^2-
        \frac{1}{9}\right)^2\pm\sqrt{\alpha^2-
        4\left(\mu^2-\frac{1}{9}\right)^3}+
      \frac{3}{4}\left(\mu^2-\frac{1}{9}\right)E^2-\frac{3}{32}E^4\right].
  \end{split}
\end{equation}
The domain of parameters for which solution~\eqref{eq:sol} is real can
be easily determined from the above formulae.

Since the system~\eqref{eq:system 1} posses two first
integrals~\eqref{eq:bloch} and~\eqref{eq:W} we are able to reduce the
dimension of the phase space to three. Hence in order to get quickly
insight into the dynamics of the considered system we make the
Poincar\'e cross sections, which are shown in
Figs.~\ref{fig:kus1}-\ref{fig:kus2}. The black bullets denotes the
elliptic solutions described enough.  Although the Poincar\'e cross
section is a standard tool for study chaotic dynamics, we did not find
its application in previous study of the Jaynes--Cummings system.

 \begin{figure}[h!]
   \centering
   \subfigure[$\mu=\frac{2}{5}, \ \alpha=\frac{44}{1875}, \
   h=\frac{44}{1875}$ ]{
     \includegraphics[width=0.48\textwidth]{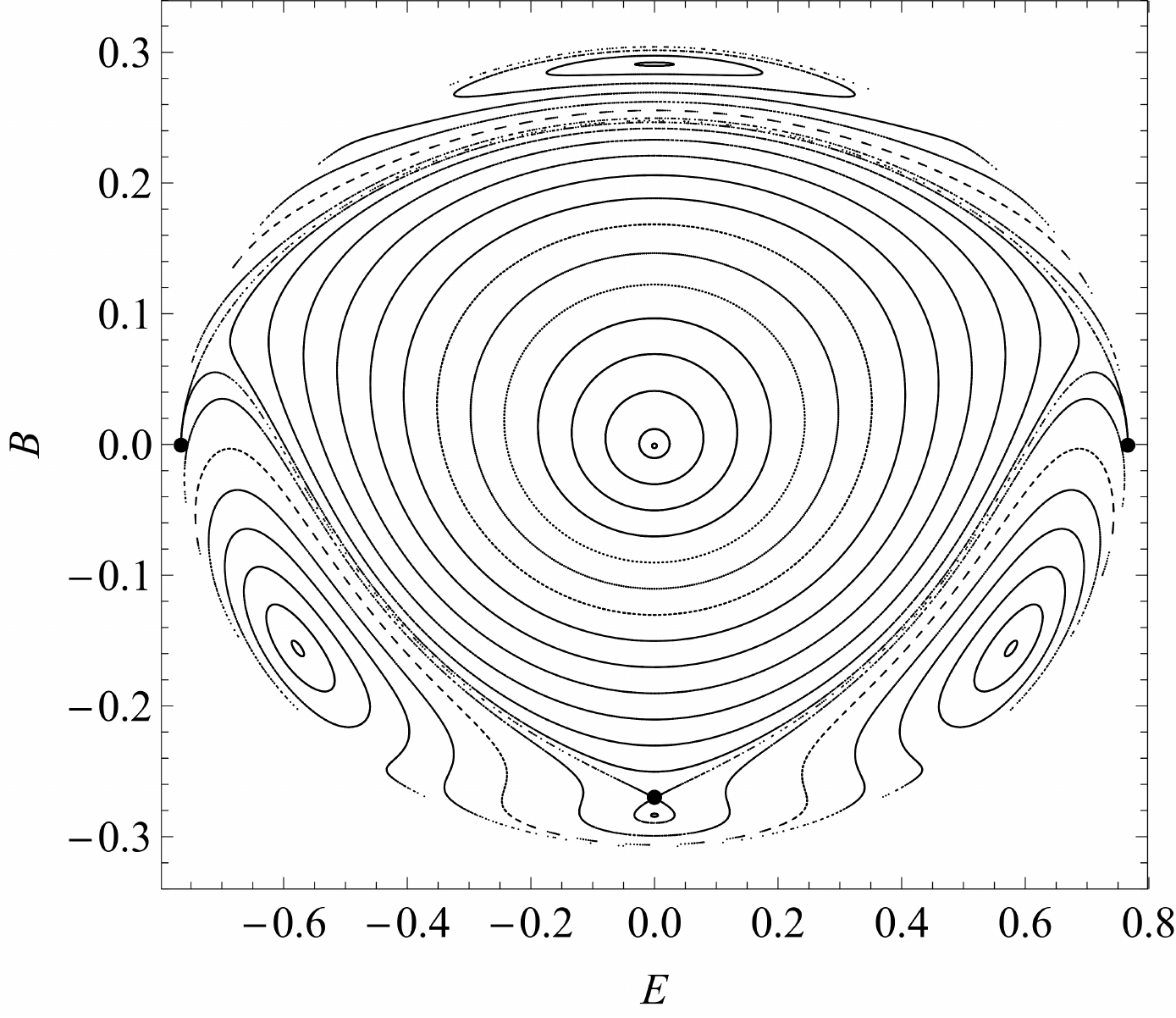}
   }
   \subfigure[$\mu=\frac{2}{3}, \ \alpha=0.43, \
   h=\frac{1}{540}\left(40-\sqrt{29769}\right)$ ]{
     \includegraphics[width=0.48\textwidth]{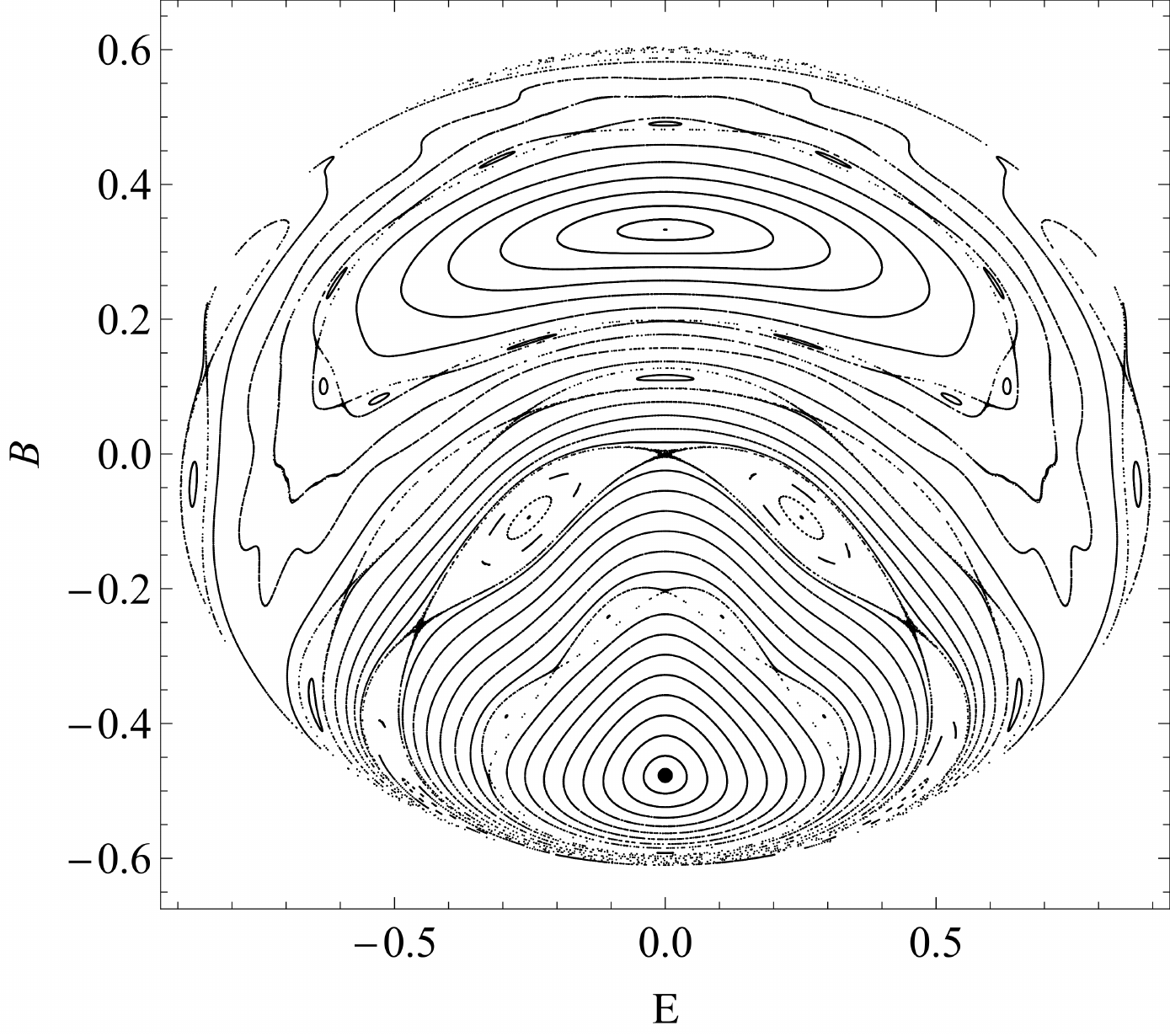}
   }
   \caption{Poincar\'e sections on the surface $x=0$ with $y>0$ and
     $z\in \R$ restricted to the plane~$(E,B)$ \label{fig:kus1}}
 \end{figure}
 \begin{figure}[h!]
   \centering
   \subfigure[$\mu=\frac{1}{2}, \ \alpha=0.42, \
   h=\frac{55}{648}-\frac{\sqrt{64759}}{1620}$ ]{
     \includegraphics[width=0.48\textwidth]{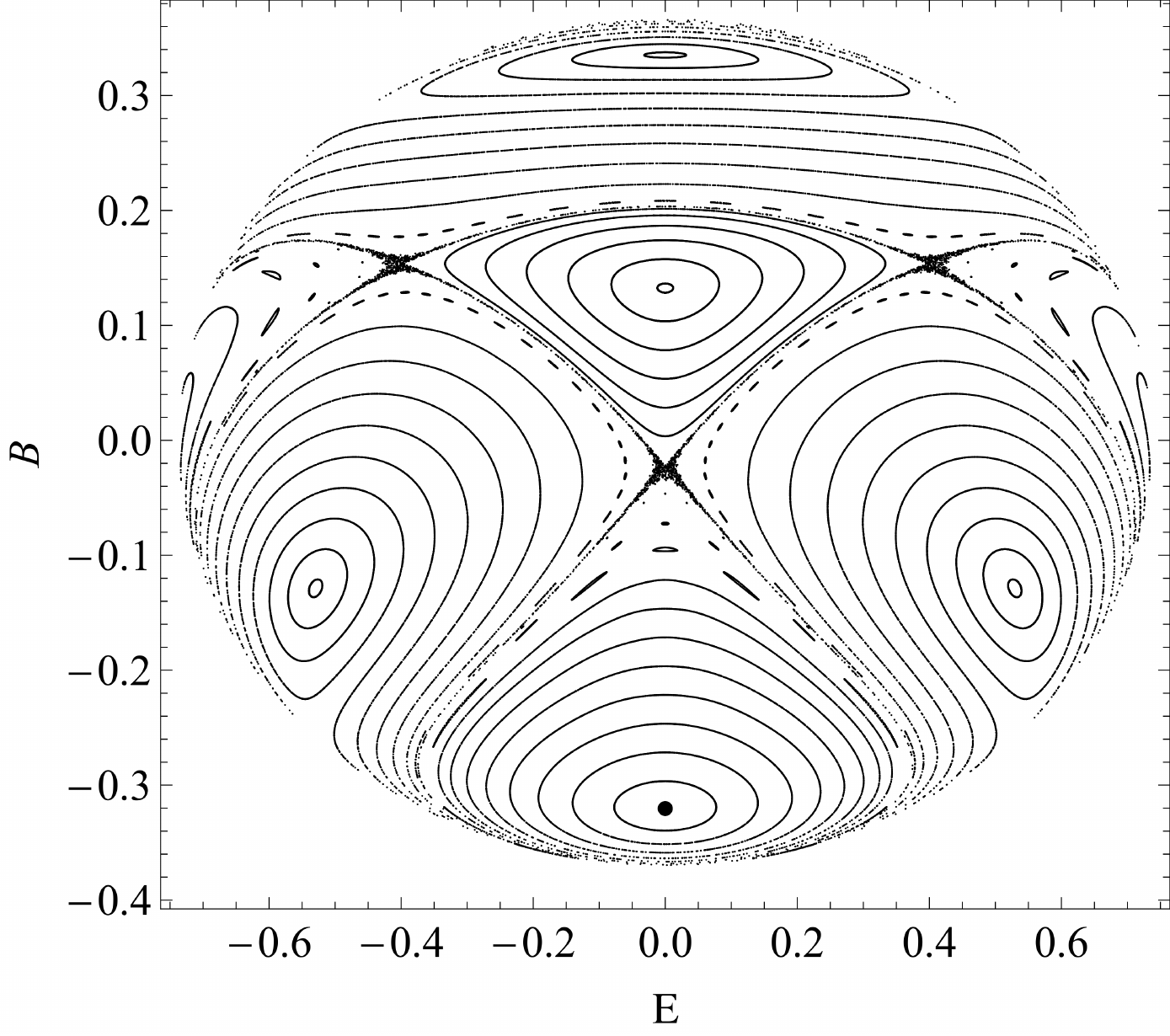}
   }
   \subfigure[$\mu=\frac{\sqrt{3}}{3}, \ \alpha=\frac{1}{4}, \
   h=\frac{8}{81}-\frac{5\sqrt{217}}{324}$ ]{
     \includegraphics[width=0.48\textwidth]{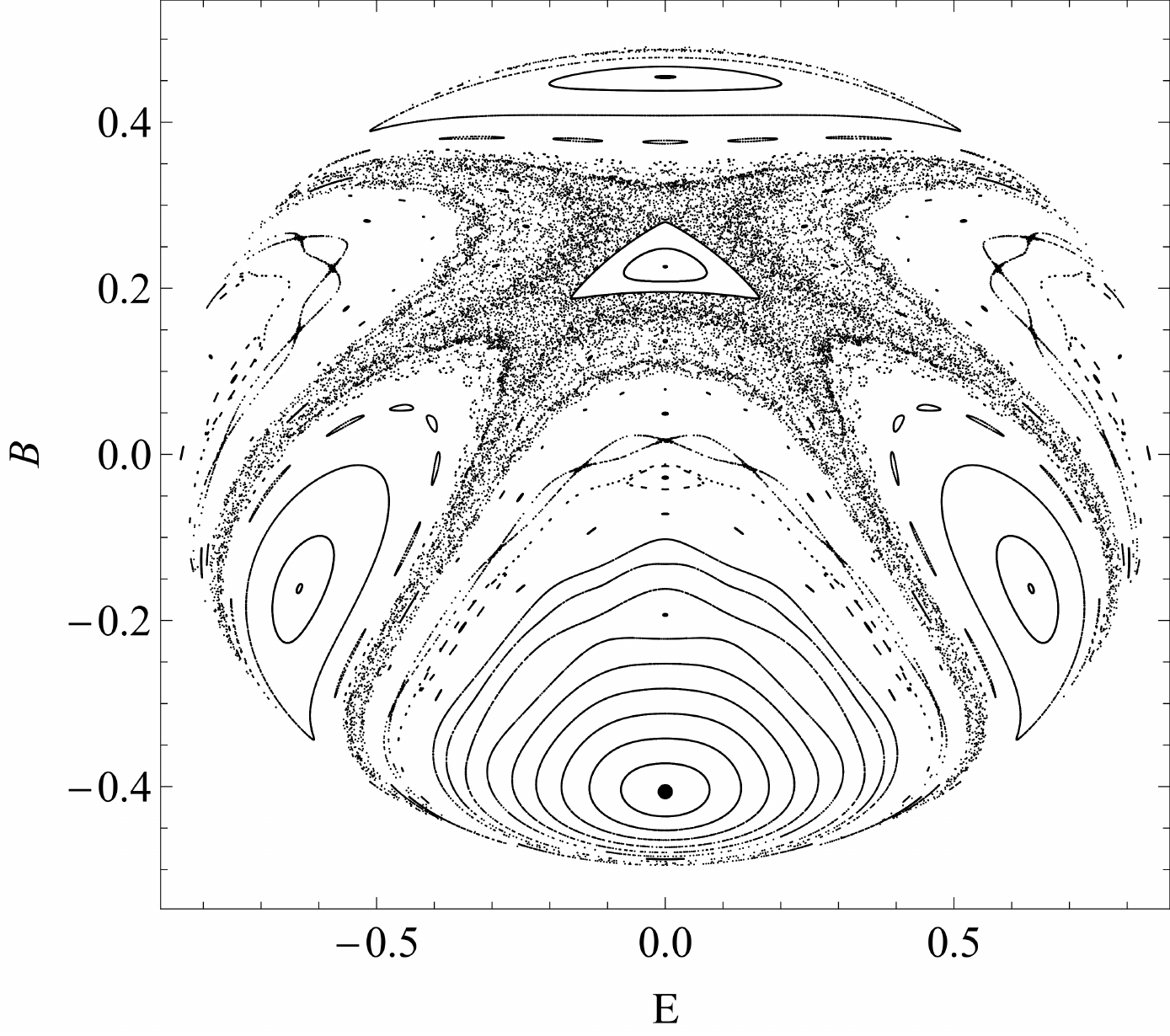}
   }
   \caption{Poincar\'e sections on the surface $x=0$ with $y>0$ and
     $z\in \R$ restricted to the plane~$(E,B)$ \label{fig:kus2}}
 \end{figure}

 Here we would like to point out the system~\eqref{eq:system 1} is
 Hamiltonian with respect to a certain degenerated Poisson structure
 in $\R^5$. To see this let us assume that $\alpha\neq 0$.  Then, we
 can rewrite system~\eqref{eq:system 1} in the following form
 
 \begin{equation}
   \label{eq:1}
   \dot \vq = \vJ(\vq) H'(\vq), \qquad \vq=[ x,y,z, E, B]^T,
 \end{equation}
 where
 \begin{equation}
   \label{eq:2}
   \vJ(\vq):= \frac{1}{\alpha} 
   \begin{bmatrix}
     0  & z & -y & 0 & 0 \\
     -z & 0 & x  & 0 & 0 \\
     y & -x & 0  & 0 & 0 \\
     0 & 0 & 0  & 0 & \alpha \\
     0 & 0 & 0 & -\alpha & 0
   \end{bmatrix}
 \end{equation}
 The Poisson bracket defined by $\vJ(\vq)$ is degenerated. First
 integral $F(\vq)$ is the Casimir of this bracket.

 If $\alpha=0$ then system~\eqref{eq:system 1} in can be written the
 following form
 \begin{equation}
   \label{eq:3}
   \dot \vq = \vJ_0(\vq) K'(\vq), \qquad K(\vq)= \frac{1}{2}F(\vq) + H(\vq)
 \end{equation}
 where
 \begin{equation}
   \label{eq:4}
   \vJ_0(\vq):= 
   \begin{bmatrix}
     0  & -1 & 0 & 0 & 0 \\
     1 & 0 & E  & 0 & 0 \\
     0& -E & 0  & 0 & 0 \\
     0 & 0 & 0  & 0 & 1 \\
     0 & 0 & 0 & -1 & 0
   \end{bmatrix}
 \end{equation}
 Although this matrix is antisymmetric it does not define the Poisson
 structure as the bracket which it defines does not satisfied the
 Jacobi identity.

 The different version of the system~\eqref{eq:system 1} has been
 studied by Miloni and co-workers, see~\cite{Milonni:83::}.  Namely
 \begin{subequations}
   \label{eq:system 2}
   \begin{align}
     \dot x &=-y  \\
     \dot y &=x +z E ,\\
     \dot z &=-y E , \\
     \dot E &=B ,\\
     \dot B &=\alpha(x+z E) -\mu^2E ,
   \end{align}
 \end{subequations}
 where the dimensionless parameters $\mu, \alpha$ have the same
 meaning as in the first case. Authors have presented the complexity
 of the system by means of Fourier analysis as well as by maximal
 Lyapunov exponent.

 We found that this version of the Jaynes--Cummings model has also two
 constants of motion
 \begin{equation}
   F=x^2+y^2+z^2,\qquad
   H_{\mathrm{M}}=\frac{1}{2}\alpha^2y^2-\alpha y B+\alpha\mu^2 z+\frac{1}{2}\mu^2E^2+\frac{1}{2}B^2.
 \end{equation}
 The existence of the first integral $H_{\mathrm{M}}$ was not reported
 in it earlier studies of this system. Although it looks quite similar
 to system \eqref{eq:system 1} we were not able to find any particular
 solution on the sphere $x^2+y^2+z^2=1$.  More importantly, it is
 seems this system cannot be written as a Hamiltonian system with
 respect to a degenerated Poisson structure. Nevertheless, we have
 found that it can be put in the following form
  \begin{equation}
   \label{eq:pm}
   \dot \vq = \vP(\vq) H'_{\mathrm{M}}(\vq), 
 \end{equation}
 where
 \begin{equation}
   \label{eq:44}
   \vP(\vq):= \frac{1}{\alpha\mu^2}
   \begin{bmatrix}
     0  & 0 & -y & 0 & 0 \\
     0 & 0 & x  & \alpha z & 0 \\
     y& -x & 0  & -\alpha y & -\alpha x \\
     0 & -\alpha z & \alpha y & 0 & -\alpha(\alpha z -\mu^2)\\
     0 & 0 & \alpha x  & \alpha(\alpha z -\mu^2) & 0
   \end{bmatrix}
 \end{equation}
 
 Figure~\ref{fig:mil1} and~\ref{fig:mil2} show the complexity of the
 system. We made the Poincar\'e cross sections for the exact
 resonance, when $\mu=1$ and $\alpha=1$ for gradually increasing value
 of energy. As we can see, for higher and higher values of energy the
 dynamics becomes more and more complex, see especially
 Figure~\ref{fig:mil2} presenting the high stage of chaotic behaviour.
 \begin{figure}[h!]
   \centering \subfigure[$\mu=1, \ \alpha=1, \ h=-0.5$ ]{
     \includegraphics[width=0.48\textwidth]{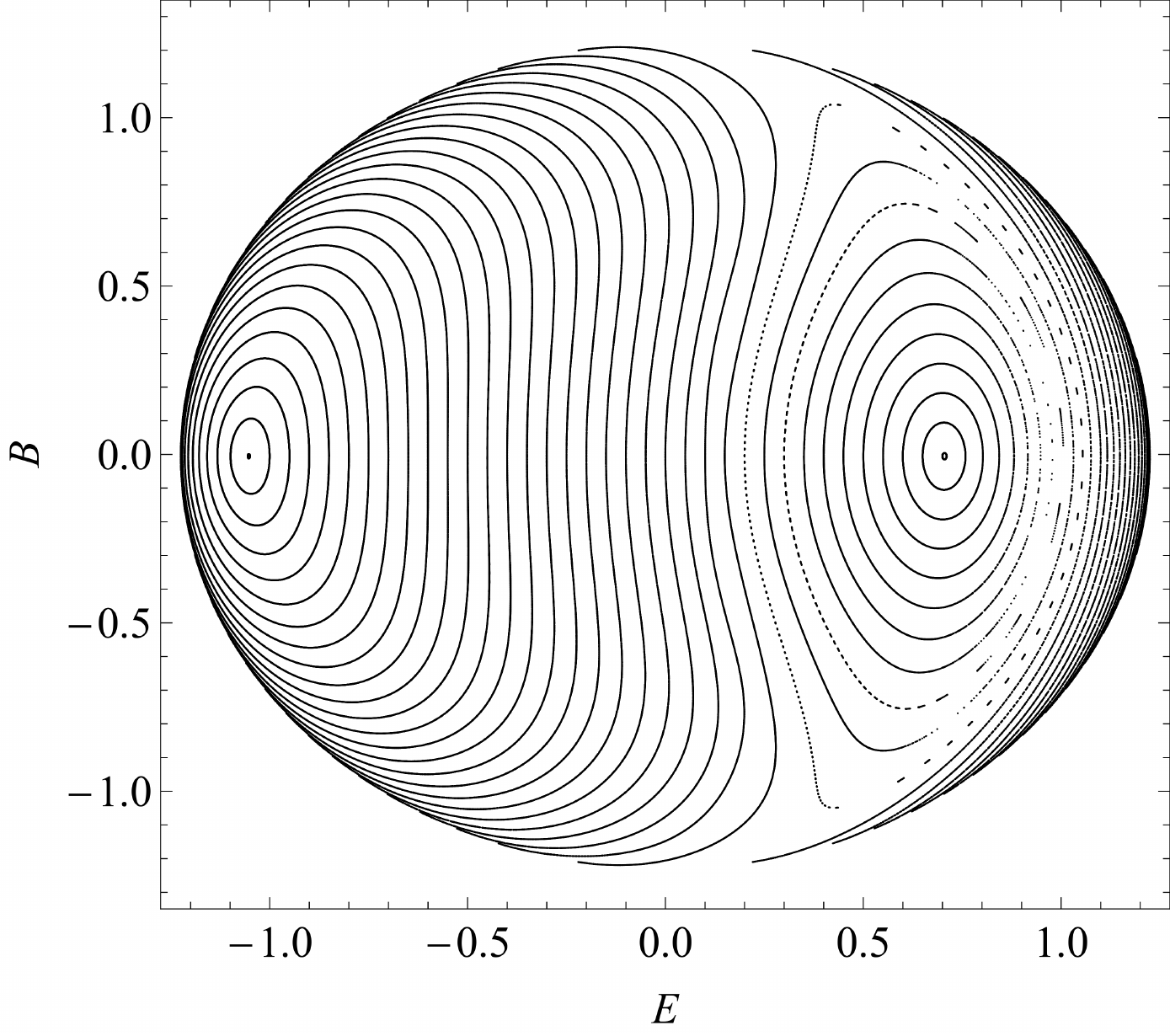}
   } \subfigure[$\mu=1, \ \alpha=1, \ h=1.2$ ]{
     \includegraphics[width=0.48\textwidth]{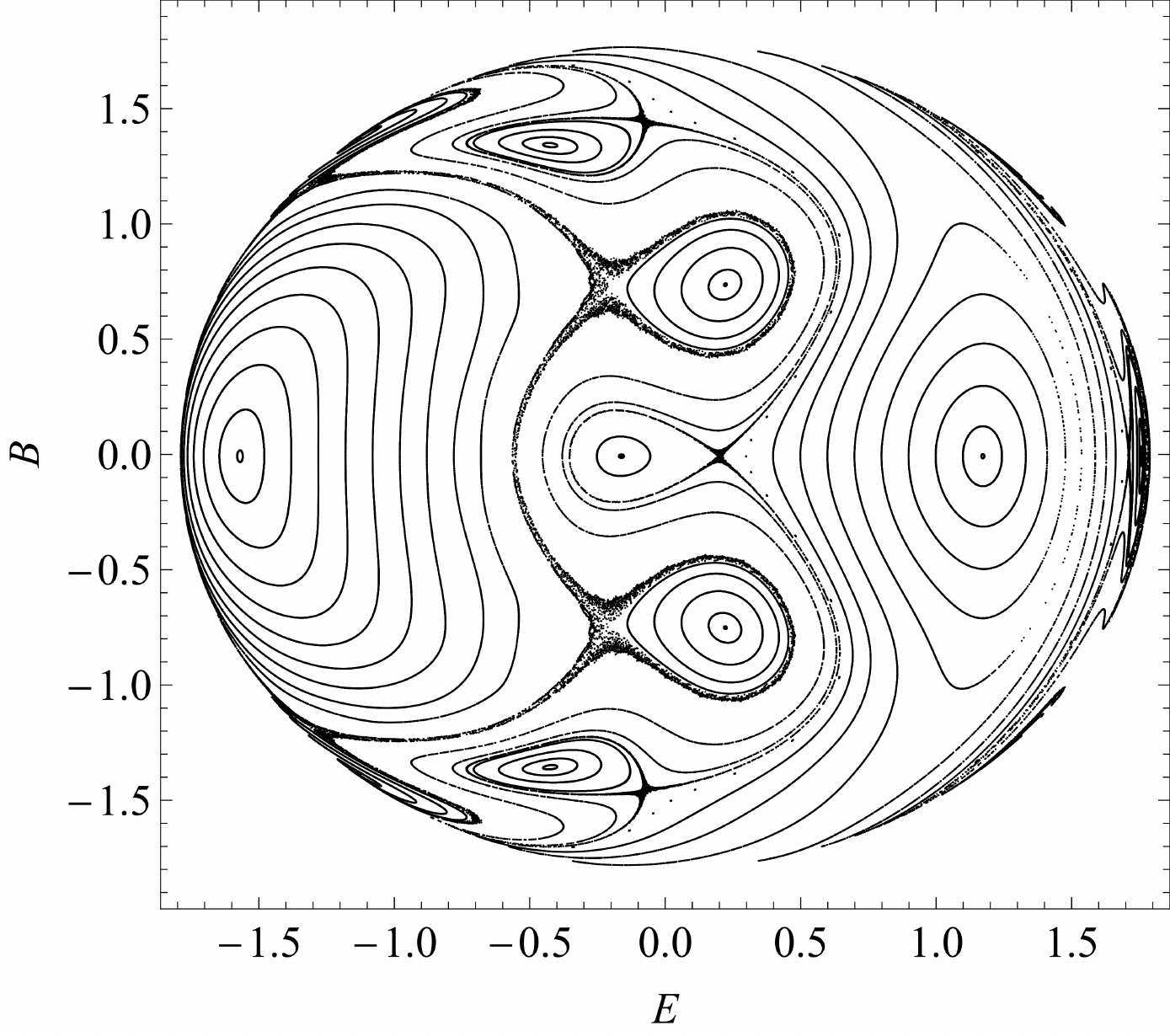}
   }
   \caption{Poincar\'e sections on the surface $x=0$ with $y>0$ and
     $z\in \R$ restricted to the plane~$(E,B)$ \label{fig:mil1}}
 \end{figure}\begin{figure}[h!]
   \centering \subfigure[$\mu=1, \ \alpha=1, \ h=1.3$ ]{
     \includegraphics[width=0.48\textwidth]{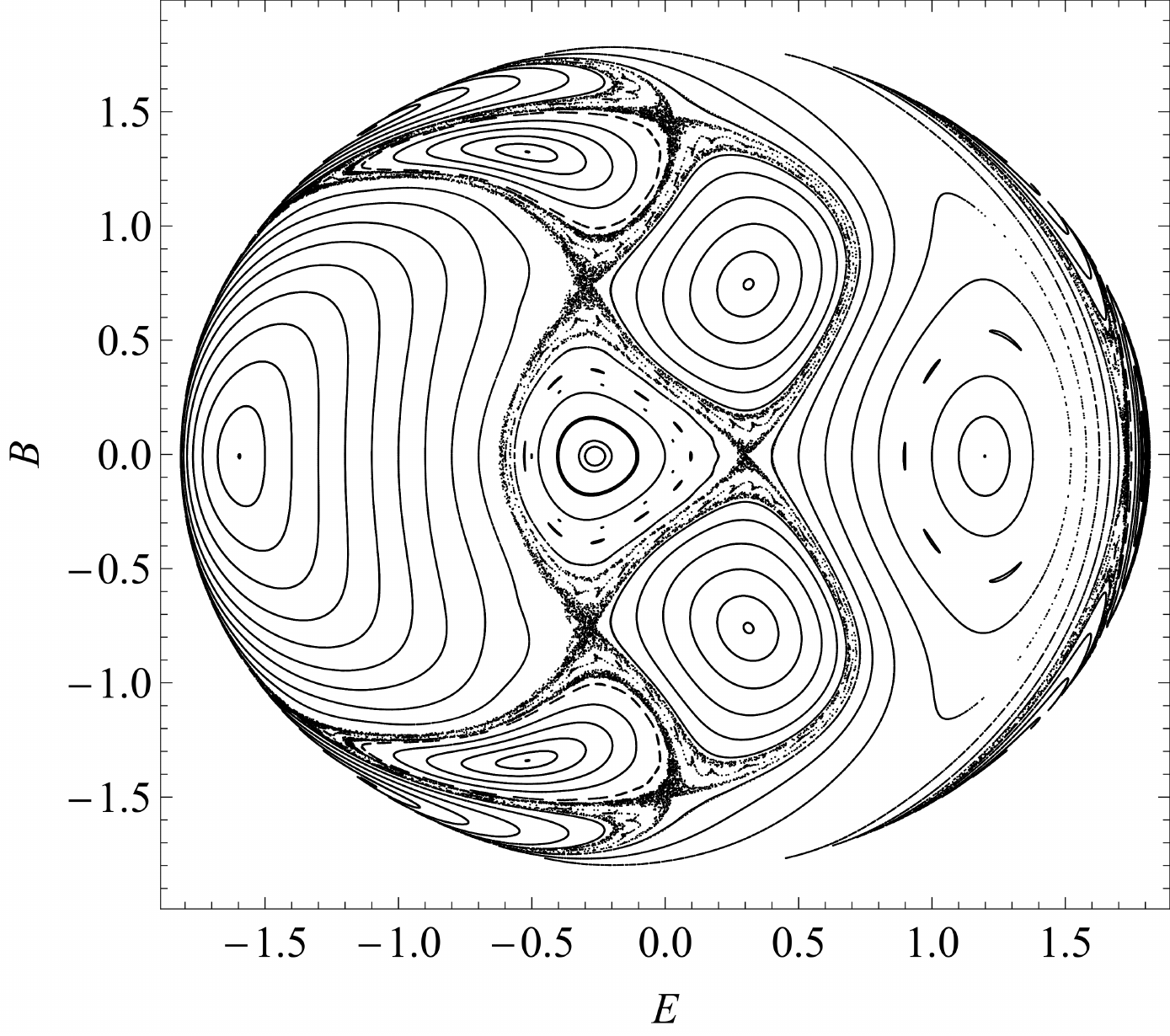}
   } \subfigure[$\mu=1, \ \alpha=1, \ h=2$ ]{
     \includegraphics[width=0.48\textwidth]{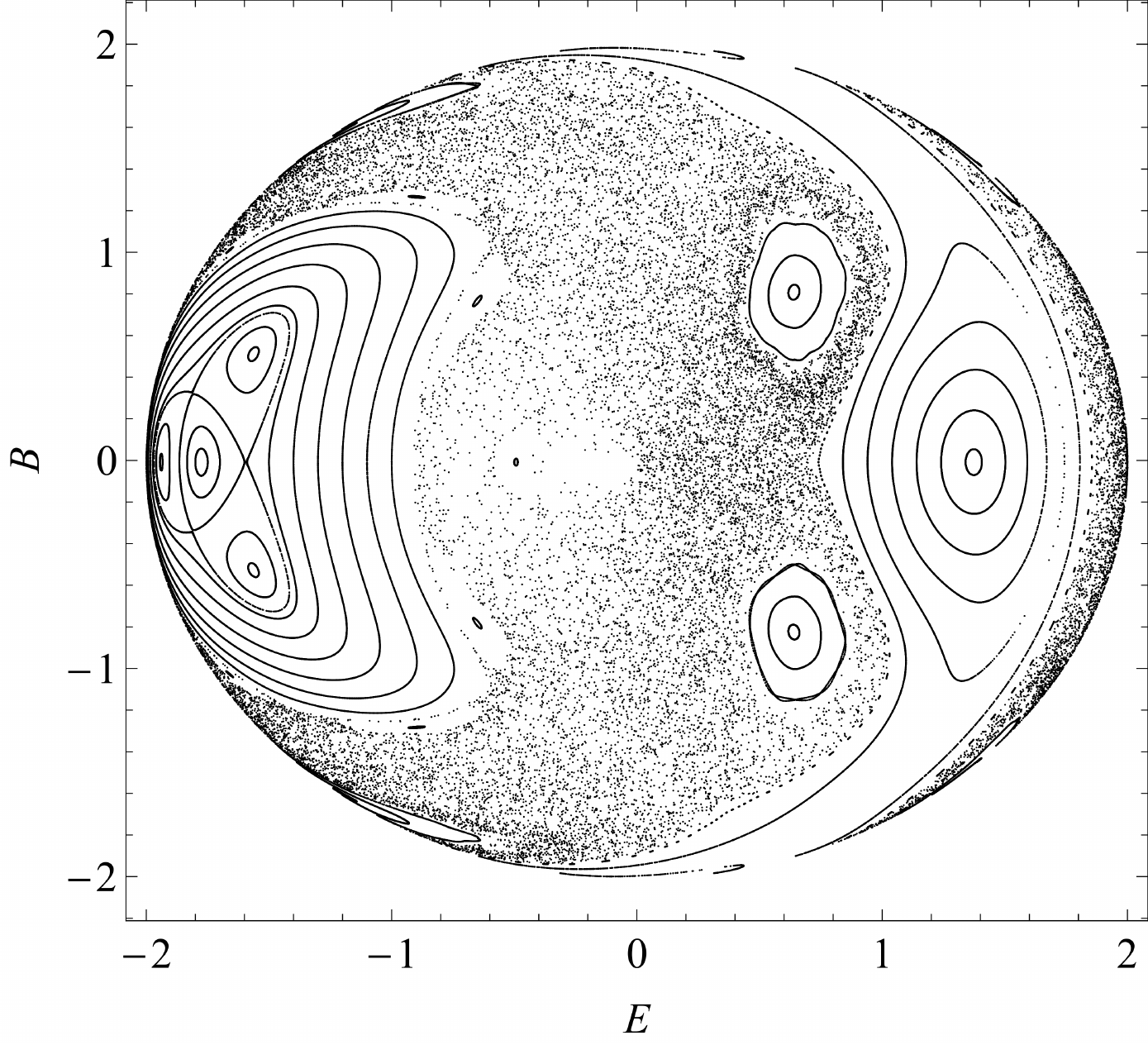}.  }
   \caption{Poincar\'e sections on the surface $x=0$ with $y>0$ and
     $z\in \R$ restricted to the plane~$(E,B)$ \label{fig:mil2}}
 \end{figure}
 As we see, the numerical analysis shows that the dynamics of the both
 systems~\eqref{eq:system 1} and ~\eqref{eq:system 2} is very complex
 and in fact chaotic. Thus, our main goal of this article is to prove
 that these systems are not integrable. More precisely, we show that
 they do not admit additional meromorphic first integral. The main result is formulated in the following theorem.
 \begin{theorem}
   \label{th:1}
   For $\mu\neq 0$, systems~\eqref{eq:system 1} and ~\eqref{eq:system
     2} do not have an additional meromorphic first which is
   functionally independent with the respective known two first
   integrals.
 \end{theorem}
 \section{Proof of Theorem~\ref{th:1}}
 Both systems~\eqref{eq:system 1} and~\eqref{eq:system 2} have the
 same invariant manifold defined by
 \begin{equation}
   \scN=\left\{(x,y,z,E,B)\in \C^5|x=y=z=0\right\}
 \end{equation}
 and its restriction to $\scN$ is defined by
 \begin{equation}
   \dot E=B,\qquad \dot B=-\mu^2E.
 \end{equation}
 \label{eq:part}
 Hence, solving equation~\eqref{eq:part} we obtain a particular
 solution
 \begin{equation}
   \label{eq:particular}
   \vvarphi(t)=(0,0,0,E(t),B(t)),\quad E(t)=\cos\mu t,\quad B(t)=-\mu\sin\mu t,
 \end{equation}
 Let $(\xi_1,\xi_2,\xi_3,\eta_1,\eta_2)^T$ be the variation of
 $(x,y,z,E,B)^T$, then the variational equations along the particular
 solution~\eqref{eq:particular} for system~\eqref{eq:system 1} takes
 the following matrix form
 \begin{equation}
   \label{eq:variational}
   \begin{pmatrix}
     \dot \xi_1 \\
     \dot \xi_2 \\
     \dot \xi_3\\
     \dot \eta_1\\
     \dot \eta_2\\
   \end{pmatrix}=
   \begin{pmatrix}
     0 & -1 & 0 & 0 & 0 \\
     1 & 0 & \cos\mu t & 0 & 0 \\
     0 & -\cos\mu t & 0 & 0 & 0 \\
     0 & 0 & 0 & 0 & 1 \\
     \alpha  & 0 & 0 & -\mu ^2 & 0 \\
   \end{pmatrix}\begin{pmatrix}
     \xi_1\\
     \xi_2\\
     \xi_3\\
     \eta_1\\
     \eta_2
   \end{pmatrix}.
 \end{equation}
 Since the motion takes place on the $(E,B)$ plane, equations for
 $(\xi_1,\xi_2,\xi_2)$ form a subsystem of the normal variational
 equations
 \begin{equation}
   \label{eq:normal}
   \dot \xi_1=-\xi_2,\quad \dot \xi_2=\xi_1+\xi_3\cos\mu t,\quad \dot \xi_3=-\xi_2\cos\mu t.
 \end{equation}
 The normal variation equations for \eqref{eq:system 2} are exactly
 the same. It is easy to check that equations~\eqref{eq:normal} have
 first integral $f=\xi_1^2+\xi_2^2+\xi_3^2=1$.  We restrict them to
 the unit sphere $\xi_1^2+\xi_2^2+\xi_3^2=1$. To this end we introduce
 new coordinates
 \begin{equation}
   u_1=\frac{\xi_3+1}{\xi_1-\rmi \xi_2},\quad u_2=-\frac{\xi_1+\rmi\xi_2}{\xi_3+1},
 \end{equation}
 so
 \begin{equation}
   \xi_1=\frac{1-u_1u_2}{u_1-u_2},\quad \xi_2=\rmi\frac{1+u_1u_2}{u_1-u_2},\quad \xi_3=\frac{u_1+u_2}{u_1-u_2}.
 \end{equation}
 Functions $u_1$ and $u_2$ satisfy the Riccati equation
 \begin{equation}
   \label{eq:Riccati}
   \dot u=A+Bu+Cu^2,
 \end{equation}
 with the coefficients
 \begin{equation}
   A=-\frac{1}{2}\rmi\cos \mu t,\quad B=\rmi,\quad C=\frac{1}{2}\rmi\cos \mu t.
 \end{equation}
 Putting $u=-\dot v/(vC)$, we transform~\eqref{eq:Riccati} to the
 linear homogeneous second-order differential equation
 \begin{equation}
   \label{eq:2nd order diff}
   \ddot v+P\dot v+Qv=0,
 \end{equation}
 with coefficients
 \begin{equation}
   P=-B-\frac{\dot C}{C}=-\rmi+\mu \tan\mu t,\quad Q=AC=\frac{1}{4}\cos^2\mu t.
 \end{equation}
 Next, by means of the change of independent variable
 \begin{equation}
   t\longrightarrow z=\exp(2\rmi\mu t),
 \end{equation}
 as well as transformation of derivatives
 \begin{equation}
   \Dt=\dot z\Dz,\qquad \Dtt= (\dot z)^2\Dzz+\ddot z\Dz,
 \end{equation}
 we transform~\eqref{eq:2nd order diff} into the equation
 \begin{equation}
   \label{eq:rational}
   v''+p(z)v'+q(z)v=0,\quad '\equiv \Dz
 \end{equation}
 with the rational coefficients
 \begin{equation}
   p(z)=-\frac{1}{1+z}+\frac{3\mu-1}{2\mu z},\quad q(z)=-\frac{(1+z)^2}{64\mu^2z^3}.
 \end{equation}
 Finally, making the classical change of dependent variable
 \begin{equation}
   \label{eq:change of dependent}
   v=w\exp\left(-\frac{1}{2}\int_{z_0}^zp(s)ds\right),
 \end{equation}
 we transform~\eqref{eq:rational} into its reduced form
 \begin{equation}
   \label{eq:reduced}
   w''=r(z)w,\qquad r(z)=-q(z)+\frac{1}{2}p'(z)+\frac{1}{4}p(z)^2,
 \end{equation}
 where the explicit form of $r(z)$ is given by
 \begin{equation}
   \label{eq:rr}
   r(z)=\frac{3}{4 (z+1)^2}+\frac{1}{64 \mu ^2 z^3}+\frac{3 \mu
     -1}{4 \mu  (z+1)}+\frac{3-4\mu-6\mu^2}{32 \mu ^2
     z^2}+\frac{1+16\mu-48\mu^2}{64 \mu ^2 z}.
 \end{equation}
 Equation~\eqref{eq:reduced} has two singular point $ z_1=0,\ z_2=1, $
 and $z_3 =\infty$. At $z=z_1$ function $r(z)$ has pole of order 3 so
 it is irregular singular point. Point $z=z_2$ is regular singularity
 of the equation and the degree of infinity is equal one. This implies
 that the differential Galois group of the equation~\eqref{eq:reduced}
 can be either dihedral or $\operatorname{SL}(2,\C)$. Hence, in order
 to check the first possibility we are going to apply the second case
 of the Kovacic algorithm.  For the detailed explanation of the
 algorithm please consult, eg.~\cite{Kovacic:86::}.
 \begin{lemma}
   The differential Galois group of the reduced
   equation~\eqref{eq:reduced} is $\operatorname{SL}(2,\C)$.
 \end{lemma}
 \begin{proof}
   For singular points $z_1$ and $z_2$ we introduce the auxiliary sets
   defined by
   \begin{equation}
     E_1=\{-2,2,6\},\quad E_2=\{3\},\quad E_\infty=\{1\}.
   \end{equation}
   Next, following the algorithm, we for for elements
   $e=(e_1,e_2,e_3)$ of the product $E=E_1\times E_2\times E_\infty$
   for which the condition
   \begin{equation}
     d(e):=\frac{1}{2}\left(e_3-e_1-e_2\right)\in \N_0
   \end{equation}
   is satisfied. It is not difficult to see that there exist only one
   element $e=\{1,-2,3\}\in E$ such that $d(e)=0$. Thus, passing
   through to the third step of the algorithm, we build the rational
   function
   \begin{equation}
     \omega(z)=\frac{1}{2}\left(\frac{e_1}{z-z_1}+\frac{e_2}{z-z_2}\right)=\frac{1}{2}\left(\frac{3}{z}-\frac{2}{z+1}\right)
   \end{equation} 
   that must satisfy the following differential equation
   \begin{equation}
     \label{eq:omega} 
     \omega''+3\omega\omega'+\omega^3-4r\omega-2r'=0,\quad '\equiv\Dz,
   \end{equation}
   where $r$ is given in~\eqref{eq:rr}. Equation above gives the
   equality
   \begin{equation}
     z^2-1+\left(2-8z+2z^2\right)\mu=0,
   \end{equation}
   that cannot be satisfied for arbitrary values of $z$. Thus, the
   differential Galois group of equation~\eqref{eq:reduced} is
   $\operatorname{SL}(2,\C)$ with non-Abelian identity component.
 \end{proof}
 In order to proof our theorem we notice that systems \eqref{eq:system
   1} and \eqref{eq:system 2} are divergence-free. Thus if they admit
 an additional first integral then, by the theorem of Jacobi, they are 
 integrable by the quadratures, see eq.  \cite{Kozlov}.  The necessary
 condition for the integrability in the Jacobi sense is that the
 identity component of the differential Galois group of the
 variational equation is Abelian, see \cite{mp:08::a}. We proved that it is
 $\operatorname{SL}(2,\C)$. Hence the systems are not integrable.
 \section*{Acknowledgement}
The work has been supported by grant No.DEC-2011/02/A/ST1/00208  
of National Science Centre of Poland.

\end{document}